\documentclass{article}
\usepackage[margin=1.25in]{geometry}
\usepackage[hidelinks]{hyperref}
\usepackage{parskip}
\usepackage{amsmath}
\usepackage{amssymb}
\usepackage{amsthm}
\usepackage{blkarray}
\usepackage{natbib}
\usepackage{color}
\usepackage{comment}
\usepackage{graphicx}
\usepackage{booktabs}

\newtheorem{theorem}{Theorem}

\newtheorem{lemma}[theorem]{Lemma}
\newtheorem{remark}[theorem]{Remark}

\title{On the Backus average of layers with randomly oriented elasticity tensors}
\author{
Len Bos\footnote{%
Dipartimento di Informatica, Universit\`a di Verona, Italy, {\tt leonardpeter.bos@univr.it}}\,, 
Michael A. Slawinski\footnote{%
Department of Earth Sciences, Memorial University of Newfoundland,
{\tt mslawins@mac.com}}\,, 
Theodore Stanoev\footnote{%
Department of Earth Sciences, Memorial University of Newfoundland,
{\tt theodore.stanoev@gmail.com}}
}
\date{}
\begin{document}
\newcommand{\C}{{\mathbb{C}}}
\newcommand{\R}{{\mathbb{R}}}
\newcommand{\Z}{{\mathbb{Z}}}
\maketitle
\begin{abstract}
As shown by~\citet{Backus1962}, the average of a stack of isotropic layers results in a transversely isotropic medium.
Herein, we consider a stack of layers consisting of a randomly oriented anisotropic elasticity tensor, which---one might expect---would result in an isotropic medium.
However, we show---by means of a fundamental symmetry of the Backus average---that the corresponding Backus average is only transversely isotropic and {\it not}, in general, isotropic.
In the process, we formulate, and use, a relationship between the Backus and~\citet{GazisEtAl1963} averages.
\end{abstract}
\section{Introduction}
In this paper, we investigate the~\citet{Backus1962} average of a stack of anisotropic layers, wherein the tensors are oriented randomly.
In spite of a conceptual relation between randomness and isotropy, herein, the Backus average results in a medium, whose anisotropy, even though weak, is irreducible to isotropy, regardless of increasing randomness.

Each layer is expressed by Hooke's law,
\begin{equation*}
\sigma_{ij}=\sum\limits_{k=1}^3\sum_{\ell=1}^3c_{ijk\ell}\,\varepsilon_{k\ell}
\,,\qquad
i,j=1,2,3\,,
\end{equation*}
where the stress tensor, $\sigma_{ij}$\,, is linearly related to the strain tensor,
\begin{equation*}
\varepsilon_{k\ell}:=\frac{1}{2}\left(\frac{\partial u_k}{\partial x_\ell}+\frac{\partial u_\ell}{\partial x_k}\right),
\qquad 
k,\ell=1,2,3\,,
\end{equation*}
where $u$ and $x$ are the displacement and position vectors, respectively, and
\begin{equation*}
c_{ijk\ell}=c_{jik\ell}=c_{k\ell ij}
\end{equation*}
is the elasticity tensor, which has to be positive-definite.
Under the index symmetries, this tensor has twenty-one linearly independent components, and can be written as \citep[e.g.,][expression~(2.1)]{BonaEtAl2008}
\begin{equation}
C
=
\left[\begin{array}{cccccc}
c_{1111} & c_{1122} & c_{1133} & \sqrt{2}\,c_{1123} & \sqrt{2}\,c_{1113} & \sqrt{2}\,c_{1112}\\
c_{1122} & c_{2222} & c_{2233} & \sqrt{2}\,c_{2223} & \sqrt{2}\,c_{2213} & \sqrt{2}\,c_{2212}\\
c_{1133} & c_{2233} & c_{3333} & \sqrt{2}\,c_{3323} & \sqrt{2}\,c_{3313} & \sqrt{2}\,c_{3312}\\
\sqrt{2}\,c_{1123} & \sqrt{2}\,c_{2223} & \sqrt{2}\,c_{3323} & 2\,c_{2323} & 2\,c_{2313} & 2\,c_{2312}\\
\sqrt{2}\,c_{1113} & \sqrt{2}\,c_{2213} & \sqrt{2}\,c_{3313} & 2\,c_{2313} & 2\,c_{1313} & 2\,c_{1312}\\
\sqrt{2}\,c_{1112} & \sqrt{2}\,c_{2212} & \sqrt{2}\,c_{3312} & 2\,c_{2312} & 2\,c_{1312} & 2\,c_{1212}
\end{array}\right]\,.
\label{eq:Chapman}
\end{equation}
Any elasticity tensor of this form is also positive-definite~\citep[e.g.,][]{BonaEtAl2007}.

A rotation of~$c_{ijk\ell}$\,, which is in~$\R^3$\,, expressed in terms of quaternions, is
\begin{equation*}\label{eq:A(q)}
A=A(q)=
\left[\begin{array}{ccc}
a^{2} + b^{2} - c^{2} - d^{2} & -2\,a\,d + 2\,b\,c & 2\,a\,c + 2\,b\,d \\
2\,a\,d + 2\,b\,c & a^{2} - b^{2} + c^{2} - d^{2} & -2\,a\,b + 2\,c\,d \\
-2\,a\,c + 2\,b\,d & 2\,a\,b + 2\,c\,d & a^{2} - b^{2} - c^{2} + d^{2} 
\end{array}\right]
\,,
\end{equation*}
where~$q=[\,a\,,\,b\,,\,c\,,\,d\,]$ is a unit quaternion.
The corresponding rotation of tensor~\eqref{eq:Chapman} is~\citep[diagram~(3.1)]{BonaEtAl2008}
\begin{equation*}
\widetilde{C}=\widetilde{A}\,C\widetilde{A}^{\,T}
\,,
\end{equation*}
where~$\tilde{A}$ is expression~\eqref{eq:At} in Appendix~\ref{app:Atilde}.
\section{The Backus and Gazis et al. averages}
\label{sec:BackusGazis}
To examine the elasticity tensors,~$C\in\R^{6\times 6}$\,, which are positive-definite, let us consider the space of all matrices~$\mathcal{M}:= \R^{6\times 6}$\,.
Its subspace of isotropic matrices is
\begin{equation*}
\mathcal{M}_{\rm iso }:=\{M\in \mathcal{M}\,:\, \widetilde{Q}\,M\,\widetilde{Q}^{\,T}=M\,,\,\,\forall\, Q\in SO(3)\}.
\end{equation*}

$\mathcal{M}_{\rm iso}$ is a linear space, since, as is easy to verify, if~$M_1,M_2\in\mathcal{M}_{\rm iso}$\,, then~$\alpha M_1+\beta M_2\in \mathcal{M}_{\rm iso}$\,, for all~$\alpha,\beta\in \R$\,.
Let us endow $\mathcal{M}$ with an inner product,
\begin{equation*}
\langle M_1,M_2\rangle_F:={\rm tr}(M_1\,M_2^{\,T})=\sum_{i,j=1}^6 (M_1)_{ij}(M_2)_{ij}\,,
\end{equation*}
and the corresponding Frobenius norm,
\begin{equation*}
\|M\|:=\sqrt{\langle M_1,M_2\rangle_F}
\,.
\end{equation*}

In such a context, \citet{GazisEtAl1963} prove the following theorem.

\begin{theorem}\label{thm:Gazis}
The closest element---with respect to the Frobenius norm---to~$M\in \mathcal{M}\subset\mathcal{M}_{\rm iso}$  is uniquely given by
\begin{equation*}
M_{\rm iso}
:=
\int\limits_{SO(3)}\widetilde{Q}\,M\,\widetilde{Q}^{\,T}\,{\rm d}\sigma(Q)\,,
\end{equation*}
where~${\rm d}\sigma(Q)$ represents the Haar probability measure on $SO(3)$\,.
\end{theorem}
\begin{proof}
It suffices to prove that
\begin{equation*}
M-M_{\rm iso}\,\,\perp \,\,\mathcal{M}_{\rm iso}\,.
\end{equation*}
To do so, we let~$N\in \mathcal{M}_{\rm iso }$ be arbitrary.
Then, for any~$A\in SO(3)$\,,
\begin{align*}
\langle
M-M_{\rm iso}\,,N
\rangle_F
&=
{\rm tr}\left(\left(M-M_{\rm iso}\right)N^{\,T}\right)
\\
&=
{\rm tr}\left(\left(
M-\int\limits_{SO(3)} \widetilde{Q}\,M\,\widetilde{Q}^{\,T}\,{\rm d}\sigma(Q)
\right)N^T\right)
\\
&=
{\rm tr}\left(\widetilde{A}\left(\left(
M - \int\limits_{SO(3)}\widetilde{Q}\,M\,\widetilde{Q}^{\,T}\,{\rm d}\sigma(Q)
\right)N^{\,T}\right)\widetilde{A}^{\,T}\right)
\tag*{{\hbox{(as~$\widetilde{A}$ is orthogonal)}}}
\\
&=
{\rm tr}\left(\widetilde{A}\,M\,N^T\widetilde{A}^{\,T} -\widetilde{A}\left(\,
\int\limits_{SO(3)}\widetilde{Q}\,M\,\widetilde{Q}^{\,T}\,{\rm d}\sigma(Q)
\right)\widetilde{A}^{\,T}\left(\widetilde{A}\,N^{\,T}\widetilde{A}^{\,T}\right)\right)\,.
\end{align*}
But
\begin{align*}
\widetilde{A}
\left(\,
\int\limits_{SO(3)}\widetilde{Q}\,M\,\widetilde{Q}^{\,T}\,{\rm d}\sigma(Q)
\right)\widetilde{A}^{\,T}
&= 
\int\limits_{SO(3)}\widetilde{A}\left(
\widetilde{Q}\,M\,\widetilde{Q}^{\,T}
\right)\widetilde{A}^{\,T}\,{\rm d}\sigma(Q)
\quad
\hbox{(by linearity)}
\\
&=
\int\limits_{SO(3)}\left(\widetilde{A}\,\widetilde{Q}\right)M \left(\widetilde{A}\,\widetilde{Q}\right)^T{\rm d}\sigma(Q)\\
&=
\int\limits_{SO(3)}\left(\widetilde{A\,Q\,}\right)M\left(\widetilde{A\,Q\,}\right)^T\,{\rm d}\sigma(Q)
\quad
\hbox{(by the properties of the tilde operation)}
\\
&=
\int\limits_{SO(3)}\widetilde{Q}\,M\,\widetilde{Q}^{\,T}\,{\rm d}\sigma(Q)
\quad
\hbox{(by the invariance of the measure)}
\\
&=
M_{\rm iso}\,.
\end{align*}
Hence,
\begin{align*}
\langle  M-M_{\rm iso}\,,N\rangle_F
&=
{\rm tr}\left(\widetilde{A}\,M\,N^T\widetilde{A}^{\,T}\right)-
M_{\rm iso}\left(\widetilde{A}\,N^T\widetilde{A}^{\,T}\right)\\
&=
{\rm tr}\left(
\left(\widetilde{A}\,M\,\widetilde{A}^{\,T}\right)\left(\widetilde{A}\,N^{\,T}\widetilde{A}^{\,T}\right)-
M_{\rm iso}\left(\widetilde{A}\,N^{\,T}\widetilde{A}^{\,T}\right)
\right)\\
&=
{\rm tr}\left(
\left(\widetilde{A}\,M\,\widetilde{A}^{\,T}\right)N^{\,T} -
M_{\rm iso}\,N^{\,T}
\right)\\
&=
{\rm tr}\left(\left(\widetilde{A}\,M\,\widetilde{A}^{\,T}\right)N^{\,T}\right) -
{\rm tr}\left(M_{\rm iso}\,N^{\,T}\right)
\,,
\end{align*}
as by assumption,~$N\in \mathcal{M}_{\rm iso}$\,.

Finally, integrating over~$A\in SO(3)$\,, we obtain
\begin{align*}
\langle  M-M_{\rm iso}\,,N\rangle_F
&=
{\rm tr}\left(\left(\,\int\limits_{SO(3)}
\widetilde{A}\,M\,\widetilde{A}^T\,{\rm d}\sigma(A)
\right)N^{\,T}\right)-
{\rm tr}(M_{\rm iso}\,N^{\,T})\\
&=
{\rm tr}(M_{\rm iso}\,N^{\,T}) - {\rm tr}(M_{\rm iso}\,N^{\,T})\\
&=0\,,
\end{align*}
as required.
\end{proof}
Since any elasticity tensor, $C\in\R^{6\times6}$\,, is positive-definite, it follows that
\begin{equation*}
C_{\rm iso} =
\int\limits_{SO(3)}\widetilde{Q}\,C\,\widetilde{Q}^{\,T}\,{\rm d}\sigma(Q)
\end{equation*}
is both isotropic and positive-definite, since it is the sum of positive-definite matrices~$\widetilde{Q}\,C\,\widetilde{Q}^{\,T}$\,.
Hence,~$C_{\rm iso}$ is the closest isotropic tensor to~$C$\,, measured in the Frobenius norm.

If $Q_i\in SO(3)$\,,~$i=1\,,\,\ldots\,,\,n$\,, is a sequence of random samples from~$SO(3)$\,, then the sample means converge almost surely to the true mean,
\begin{equation}\label{eq:SampleMeans}
\lim_{n\to\infty}\frac{1}{n}\sum_{i=1}^n \widetilde{Q_i}\,C\, \widetilde{Q_i}^{\,T}
=
\int\limits_{SO(3)} \widetilde{Q}\,C\,\widetilde{Q}^{\,T}\,{\rm d}\sigma(Q)=C_{\rm iso}\,,
\end{equation}
which---in accordance with Theorem~\ref{thm:Gazis}---is the Gazis et al. average of $C$\,.

This paper relies on replacing the arithmetic average in expression~\eqref{eq:SampleMeans} by the Backus average, which provides a single, homogeneous model that is long-wave-equivalent to a thinly layered medium.
According to~\citet{Backus1962}, the average of the function~$f(x_3)$ of ``width''~$\ell'$ is the moving average given by
\begin{equation*}
\overline{f}(x_3):=\int\limits_{-\infty}^\infty w(\zeta-x_3)f(\zeta)\,{\rm d}\zeta
\,,
\end{equation*}
where the weight function,~$w(x_3)$\,, acts like the Dirac delta centred at~$x_3=0$\,, and exhibits the following properties.
\begin{equation*}
w(x_3)\geqslant0\,,\!\!
\quad w(\pm\infty)=0\,,\!\!
\quad
\int\limits_{-\infty}^\infty w(x_3)\,{\rm d}x_3=1\,,\!\!
\quad
\int\limits_{-\infty}^\infty x_3\,w(x_3)\,{\rm d}x_3=0\,,\!\!
\quad
\int\limits_{-\infty}^\infty x_3^2\,w(x_3)\,{\rm d}x_3=(\ell')^2\,.
\end{equation*}
These properties define~$w(x_3)$ as a probability-density function with mean zero and standard deviation~$\ell'$\,, thus explaining the term ``width'' for~$\ell'$\,.
\section{The block structure of $C$ $\to$ $\widetilde{A}\,C\widetilde{A}^{\,T}$}
The action $C\to \widetilde{A}\,C\widetilde{A}^{\,T}$ has a simple block structure that is exploited in Section~\ref{sec:FundamentalSymmetry}.
To see this, we consider $q=[\,a\,,\,0\,,\,0\,,\,d\,]$\,, with~$a:=\cos(\theta/2)$\,,~$d:=\sin(\theta/2)$\,; thus, in accordance with expression~\eqref{eq:A(q)},
\begin{equation}\label{eq:A1}
A=A(q)=\left[\begin{array}{ccc}
\cos\theta & -\sin\theta & 0 \\
\sin\theta & \cos\theta & 0 \\
0 & 0 & 1
\end{array}\right]
\end{equation}
and, in accordance with expression~\eqref{eq:At},
\begin{equation*}
\widetilde{A}=
\left[\def\arraystretch{1.5}\begin{array}{c*{5}{c}}
\cos^{2}\theta & \sin^{2}\theta & 0 & 0 & 0 & -\tfrac{1}{\sqrt{2}}\sin\left(2\,\theta\right) \\
\sin^{2}\theta & \cos^{2}\theta & 0 & 0 & 0 & \tfrac{1}{\sqrt{2}}\sin\left(2\,\theta\right) \\
0 & 0 & 1 & 0 & 0 & 0 \\
0 & 0 & 0 & \cos\theta & \sin\theta & 0 \\
0 & 0 & 0 & -\sin\theta & \cos\theta & 0 \\
\tfrac{1}{\sqrt{2}}\sin\left(2\,\theta\right) & -\tfrac{1}{\sqrt{2}}\sin\left(2\,\theta\right) & 0 & 0 & 0 & \cos\left(2\,\theta\right)
\end{array}\right]\,.
\end{equation*}
For~$q=[\,0\,,\,b\,,\,c\,,\,0\,]$\,, with~$b:=\cos(\theta/2)$ and~$c:=\sin(\theta/2)$\,,
\begin{equation}\label{eq:A2}
A=A(q)=\left[\begin{array}{ccc}
\cos\theta & \sin\theta & 0 \\
\sin\theta & -\cos\theta & 0 \\
0 & 0 & -1
\end{array}\right] 
\end{equation}
and
\begin{equation*}\widetilde{A}=
\left[\def\arraystretch{1.5}\begin{array}{c*{5}{c}}
\cos^{2}\theta & \sin^{2}\theta & 0 & 0 & 0 & \tfrac{1}{\sqrt{2}}\sin\left(2\,\theta\right) \\
\sin^{2}\theta & \cos^{2}\theta & 0 & 0 & 0 & -\tfrac{1}{\sqrt{2}}\sin\left(2\,\theta\right) \\
0 & 0 & 1 & 0 & 0 & 0 \\
0 & 0 & 0 & \cos\theta & -\sin\theta & 0 \\
0 & 0 & 0 & -\sin\theta & -\cos\theta & 0 \\
\tfrac{1}{\sqrt{2}}\sin\left(2\,\theta\right) & -\tfrac{1}{\sqrt{2}}\sin\left(2\,\theta\right) & 0 & 0 & 0 & -\cos\left(2\,\theta\right)
\end{array}\right]
\,.
\end{equation*}
In both cases permuting the rows and columns to the order~$(\,3\,,\,4\,,\,5\,,\,1\,,\,2\,,\,6\,)$ results in a diagonal block structure for~$\widetilde{A}$\,.
For expression~\eqref{eq:A1}, we have
\begin{equation*}
\widetilde{A} 
\to 
\left[\begin{array}{cc} 
\widetilde{A}_1 & 0 \\
0 & \widetilde{A}_2 
\end{array}\right]
\,,
\end{equation*}
where
\begin{equation*}
\widetilde{A}_1
=
\left[\begin{array}{ccc}
1 & 0 & 0 \\
0 & \cos\theta & \sin\theta \\
0 & -\sin\theta & \cos\theta
\end{array}\right]
\quad\hbox{and}\quad
\widetilde{A}_2=
\left[\def\arraystretch{1.5}\begin{array}{ccc}
\cos^{2}\theta & \sin^{2} & -\tfrac{1}{\sqrt{2}}\sin\left(2\,\theta\right) \\
\sin^{2}\theta & \cos^{2} & \tfrac{1}{\sqrt{2}}\sin\left(2\,\theta\right) \\
\tfrac{1}{\sqrt{2}}\sin\left(2\,\theta\right) & -\tfrac{1}{\sqrt{2}}\sin\left(2\,\theta\right) & \cos\left(2\,\theta\right)
\end{array}\right].
\end{equation*}
Both~$\widetilde{A}_1,\widetilde{A}_2\in\R^{3\times3}$ are rotation matrices.
Similarly, for expression~\eqref{eq:A2},
\begin{equation*}
\widetilde{A}
\to 
\left[\begin{array}{cc}
\widetilde{A}_1 & 0 \\ 
0 & \widetilde{A}_2 
\end{array}\right]
\,;
\end{equation*}
herein, $\widetilde{A}_1,\widetilde{A}_2\in\R^{3\times3}$ are reflection matrices.
Thus, in both cases, $\widetilde{A}_1$ and $\widetilde{A}_2$ are orthogonal matrices.

In either case, the following lemma holds.
\begin{lemma}\label{lem:BlockStructure}
Suppose that the rows and columns of~$C$ are permuted to the order~$(\,3\,,\,4\,,\,5\,,\,1\,,\,2\,,\,6\,)$ to have the block structure
\begin{equation*}
C
\to 
\left[\begin{array}{cc}
M & B \\ 
K & J 
\end{array}\right]\,,
\end{equation*}
with~$M\,,\,B\,,\,K\,,\,J\in\R^{3\times3}$\,, and that the rows and columns of~$\widetilde{A}$ are also so permuted.
Then,
\begin{equation*}
\widetilde{A}\,C\widetilde{A}^{\,T}\,\,\to\,\, 
\left[\def\arraystretch{1.25}\begin{array}{cc} 
\widetilde{A}_1\,M\,\widetilde{A}_1^{\,T} & \widetilde{A}_1\,B\,\widetilde{A}_2^{\,T} \cr
\widetilde{A}_2\,K\,\widetilde{A}_1^{\,T} & \widetilde{A}_2\,J\,\widetilde{A}_2^{\,T}
\end{array}\right]
\,.
\end{equation*}
\end{lemma}
\begin{proof}
Let~$P\in\R^{6\times6}$ be the matrix obtained by permuting the rows of the identity to the order~$(\,3\,,\,4\,,\,5\,,\,1\,,\,2\,,\,6\,)$\,.
Our assumption is that
\begin{equation*}
P\,C\,P^{\,T} 
=
\left[\begin{array}{cc} 
M & B \\ 
K & J 
\end{array}\right]\,.
\end{equation*}
Then,
\begin{align*}
P\left(\widetilde{A}\,C\,\widetilde{A}^{\,T}\right)P^{\,T}
&=
\left(P\widetilde{A}\,P^{\,T}\right)\left(P\,C\,P^{\,T}\right)\left(P\,\widetilde{A}^{\,T}P^{\,T}\right)\\
&=
\left[\begin{array}{cc}
\widetilde{A}_1 & 0 \\ 
0 & \widetilde{A}_2 
\end{array}\right]
\left[\begin{array}{cc}
M & B \\ 
K & J 
\end{array}\right]
\left[\begin{array}{cc}
\widetilde{A}_1^{\,T} & 0 \\ 
0 & \widetilde{A}_2^{\,T}
\end{array}\right]\\
&=
\left[\def\arraystretch{1.25}\begin{array}{cc}
\widetilde{A}_1\,M\,\widetilde{A}_1^{\,T} & \widetilde{A}_1\,B\,\widetilde{A}_2^{\,T} \\
\widetilde{A}_2\,K\,\widetilde{A}_1^{\,T} & \widetilde{A}_2\,J\,\widetilde{A}_2^{\,T}
\end{array}\right]\,,
\end{align*}
as required.
\end{proof}
\section{The fundamental symmetry of the Backus average}
\label{sec:FundamentalSymmetry}
Let us examine properties of the Backus average, which---for elasticity tensors, $C_{i}$---we denote by
\begin{equation*}
\overline{\overline{\left(\,C_{1}\,,\,\ldots\,,\,C_{n}\,\right)}}\,.	
\end{equation*}
\begin{theorem}\label{thm:SymProp}
For~$A\in\R^{3\times 3}$ of the form of expression~\eqref{eq:A1} or \eqref{eq:A2} and any elasticity tensor, $C_{1}\,,\,\dots\,,\,C_{N}$ $\in\mathbb{R}^{6\times6}$\,,
\begin{equation}\label{eq:BAsym}
\widetilde{A}\,\,\overline{\overline{
\left(C_{1}\,,\,\ldots\,,\,C_{n}\right)
}}\,\,\widetilde{A}^{\,T}
=
\overline{\overline{
\left(\widetilde{A}\,C_{1}\,\widetilde{A}^{\,T}\,,\,\dots\,,\widetilde{A}\,C_{n}\,\widetilde{A}^{\,T}\right)
}}
\,,
\end{equation}
which is a symmetry condition.
Conversely, if for an orthogonal matrix, $A\in\R^{3\times3}$\,, we have equality~\eqref{eq:BAsym}, for any collection of elasticity tensors, $C_1\,,\,\ldots\,,\,C_n\in\R^{6\times6}$\,, then $A$ must be of the form of expression~\eqref{eq:A1} or \eqref{eq:A2}.
\end{theorem}
\begin{proof}
As in Lemma~\ref{lem:BlockStructure}, we permute the rows and columns of~$C\in\mathbb{R}^{6\times6}$ to the order~$(\,3\,,\,4\,,\,5\,,\,1\,,\,2\,,\,6\,)$\,. 
Thus, we have the block structure
\begin{equation*}
C
=
\left[\begin{array}{cc}
M & B \\
K & J
\end{array}\right]
\,,\quad
M\,,\,B\,,\,K\,,\,J\in\mathbb{R}^{3\times3};
\end{equation*}
herein, we use the notation of equations (5)--(9) of~\citet{BDaltonSS2017}.
Also,~$\widetilde{A}$ has the block structure of
\begin{equation*}
\widetilde{A}
=
\left[\begin{array}{cc}
\widetilde{A}_{1} & 0 \\
0 & \widetilde{A}_{2}
\end{array}\right]
\,,\quad
\widetilde{A}_{1}\,,\,\widetilde{A}_{2}\in\mathbb{R}^{3\times3}\,,
\end{equation*}
and is orthogonal.

Let 
\begin{equation*}
\widetilde{C}
=
\widetilde{A}\,C\widetilde{A}^T 
=
\left[\begin{array}{cc}
\widetilde{A}_{1} & 0 \\
0 & \widetilde{A}_{2}
\end{array}\right]
\left[\begin{array}{cc}
M & B \\
K & J
\end{array}\right]
\left[\begin{array}{cc}
\widetilde{A}_{1}{}^T & 0 \\
0 & \widetilde{A}_{2}{}^T
\end{array}\right]
=
\left[\begin{array}{cc}
\widetilde{M} & \widetilde{B} \\
\widetilde{K} & \widetilde{J}
\end{array}\right]\,,
\end{equation*}
where, by Lemma~\ref{lem:BlockStructure},
\begin{equation*}
\widetilde{M} = \widetilde{A}_{1}\,M\,\widetilde{A}_{1}{}^T
\,,\qquad
\widetilde{B} = \widetilde{A}_{1}\,B\,\widetilde{A}_{2}{}^T
\,,\qquad
\widetilde{K} = \widetilde{A}_{2}\,K\,\widetilde{A}_{1}{}^T
\,,\qquad
\widetilde{J} = \widetilde{A}_{2}\,J\,\widetilde{A}_{2}{}^T
\,.
\end{equation*}
In particular,
\begin{align}
\label{eq:tildeMinverse}
\widetilde{M}^{-1}
&=
\left(\widetilde{A}_{1}\,M\,\widetilde{A}_{1}{}^T\right)^{-1} \\
\nonumber
&=
\widetilde{A}_{1}\,M^{-1}\,\widetilde{A}_{1}{}^T
\,.
\end{align}
The Backus-average equations are given by \citep{BDaltonSS2017}
\begin{equation}\label{eq:BAforms}
C_{BA}
=
\left[\begin{array}{cc}
M_{BA} & B_{BA}\\
K_{BA} & J_{BA}
\end{array}\right]
\,,
\end{equation}
where
\begin{align*}
M_{BA} &= \left(\,\overline{M^{-1}}\,\right)^{\!-1}\,, \\
B_{BA} &= \left(\,\overline{M^{-1}}\,\right)^{\!-1}\,\overline{M^{-1}\,B}\,, \\
K_{BA} &= \overline{K\,M^{-1}}\left(\,\overline{M^{-1}}\,\right)^{\!-1}\,,\\
J_{BA} 
&= 
\left(\overline{J} - \overline{K\,M^{-1}\,B} + \overline{K\,M^{-1}}\left(\,\overline{M^{-1}}\,\right)^{\!-1}\overline{M^{-1}\,B}\right)\,,
\end{align*}
and where~$\overline{\vphantom{\left(\right.}\hspace*{0.5pt}\circ\vphantom{\left.\right)}}$ denotes the arithmetic average of the expression~$\circ$\,; for example,
\begin{equation*}
\overline{M^{-1}}=\frac{1}{n}\sum_{i=1}^n M^{-1}_i\,.
\end{equation*}
Let $M_{BA}$\,, $B_{BA}$\,, $K_{BA}$ and $J_{BA}$ denote the associated sub-blocks of the Backus average of the~$\widetilde{A}\,C_{j}\widetilde{A}^{\,T}$\,.
Then,
\begin{align*}
\widetilde{A}_{1}\,M_{BA}\,\widetilde{A}_{1}{}^T
&=
\widetilde{A}_{1}\left(\,\overline{M^{-1}}\,\right)^{-1}\widetilde{A}_{1}{}^T &\\
&=
\left(\widetilde{A}_{1}\,\overline{M^{-1}}\widetilde{A}_{1}{}^T\right)^{\!-1} &\\
&=
\left(\,\overline{\widetilde{A}_{1}\,M^{-1}\,\widetilde{A}_{1}{}^T}\,\right)^{\!-1}
&
\left(\mathrm{by\,\,linearity}\right) \\
&=
\left(\,\overline{\left(\widetilde{A}_{1}\,M\,\widetilde{A}_{1}{}^T\right)^{-1}}\,\right)^{\!-1}
&
\left({\rm by\,\,equation~\eqref{eq:tildeMinverse}}\right) \\
&=
\left(\,\overline{\left(\widetilde{M}^{-1}\right)}\,\right)^{\!-1} \\
&=
\widetilde{M}_{BA}
\,,
\end{align*}
\begin{align*}
\widetilde{A}_{1}\,B_{BA}\,\widetilde{A}_{2}{}^T
&=
\widetilde{A}_{1}
\left(\,\left(\,\overline{M^{-1}}\,\right)^{-1}\,\overline{M^{-1}B}\,\right)
\widetilde{A}_{2}{}^T 
&
\\
&=
\left(\,\widetilde{A}_{1}\left(\,\overline{M^{-1}}\,\right)^{\!-1}\widetilde{A}_{1}{}^{T}\right)
\left(\widetilde{A}_{1}\,\overline{M^{-1}B}\,\widetilde{A}_{2}{}^T\,\right)
&
\\
&=
\widetilde{M}_{BA}\left(\widetilde{A}_{1}\,\overline{M^{-1}B}\,\widetilde{A}_{2}{}^T\right)
&
\left({\rm by~the~previous~result }\right)
\\
&=
\widetilde{M}_{BA}\,\overline{\left(\widetilde{A}_{1}\,M^{-1}B\,\widetilde{A}_{2}{}^T\right)}
&
\left({\rm by\,\,linearity}\right) 
\\
&=
\widetilde{M}_{BA}\,\overline{\left(\widetilde{A}_{1}\,M^{-1}\widetilde{A}_{1}{}^T\right)\left(\widetilde{A}_{1}\,B\,\widetilde{A}_{2}{}^T\right)}
\\
&=
\widetilde{M}_{BA}\,\overline{\widetilde{M}^{-1}\,\widetilde{B}\,}
\\
&=
\widetilde{B}_{BA}
\,,
\end{align*}
\begin{align*}
\widetilde{A}_{2}\,K_{BA}\,\widetilde{A}_{1}{}^T
&=
\widetilde{A}_{2}\,\overline{K\,M^{-1}}\left(\,\overline{M^{-1}}\,\right)^{\!-1}\!\widetilde{A}_{1}{}^T
\\
&=
\left(\widetilde{A}_{2}\,\overline{K\,M^{-1}}\,\widetilde{A}_{1}{}^T\right)
\left(\widetilde{A}_{1}\left(\overline{M^{-1}}\,\right)^{\!-1}\!\widetilde{A}_{1}{}^T\right)
\\
&=
\overline{\left(
\widetilde{A}_{2}\,K\,M^{-1}\,\widetilde{A}_{1}{}^T
\right)}\,\widetilde{M}_{BA}
\\
&=
\overline{\widetilde{K}\,\widetilde{M}^{-1}\,}\,\widetilde{M}_{BA}
\\
&=
\widetilde{K}_{BA}
\,,
\end{align*}
and
\begin{align*}
\widetilde{A}_{2}\,J_{BA}\,\widetilde{A}_{2}{}^T
&=
\widetilde{A}_{2}
\left(\overline{J\,} - \overline{K\,M^{-1}\,B} + \overline{K\,M^{-1}}\left(\,\overline{M^{-1}}\,\right)^{\!-1}\overline{M^{-1}\,B}\right)
\widetilde{A}_{2}{}^T
\\
&=
\widetilde{A}_{2}\left(\,\overline{J\,}\,\right)\widetilde{A}_{2}{}^T -
\widetilde{A}_{2}\left(\,\overline{K\,M^{-1}\,B}\,\right)\widetilde{A}_{2}{}^T +
\widetilde{A}_{2}\left(\overline{K\,M^{-1}}\left(\,\overline{M^{-1}}\,\right)^{\!-1}\overline{M^{-1}\,B}\,\right)\widetilde{A}_{2}{}^T
\\
&=
\overline{\widetilde{A}_{2}\left(\,J\,\right)\widetilde{A}_{2}{}^T} -
\overline{\widetilde{A}_{2}\left(\,K\,M^{-1}\,B\,\right)\widetilde{A}_{2}{}^T} +
\widetilde{A}_{2}\,\overline{K\,M^{-1}}\,\widetilde{A}_{1}{}^T\,\,
\widetilde{A}_{1}\left(\,\overline{M^{-1}}\,\right)^{\!-1}\!\widetilde{A}_{1}{}^T\,\,
\widetilde{A}_{1}\,\overline{M^{-1}\,B}\,\widetilde{A}_{2}{}^T
\\
&=
\overline{\widetilde{J}\,\,} -
\overline{
\left(\widetilde{A}_{2}\,K\,\widetilde{A}_{1}{}^T\right)
\left(\widetilde{A}_{1}\,M^{-1}\,\widetilde{A}_{1}{}^T\right)
\left(\widetilde{A}_{1}\,B\,\widetilde{A}_{2}{}^T\right)
} +
\\
&\qquad\quad
\overline{\widetilde{A}_{2}\,K\,M^{-1}\widetilde{A}_{1}{}^T}
\left(\widetilde{A}_{1}\left(\,\overline{M^{-1}}\,\right)^{\!-1}\!\widetilde{A}_{1}{}^T\right)
\overline{\left(
\widetilde{A}_{1}\,M^{-1}\widetilde{A}_{1}{}^T\,\,
\widetilde{A}_{1}\,B\,\widetilde{A}_{2}{}^T
\right)}
\\
&=
\overline{\widetilde{J}\,\,} -
\overline{\widetilde{K}\,\widetilde{M}^{-1}\widetilde{B}\,} +
\overline{\widetilde{K}\,\widetilde{M}^{-1}\,}\,\widetilde{M}_{BA}\,\overline{\widetilde{M}^{-1}\widetilde{B}\,}
\\
&=
\widetilde{J}_{BA}
\,,
\end{align*}
which completes the proof of equality~\eqref{eq:BAsym}.

To show the converse claimed in the statement of Theorem~\ref{thm:SymProp}, let us consider $C_{1}=I$ and $C_{2}=2\,I$\,.
Their Backus average is
\begin{equation}\label{eq:BAofC1andC2}
B
:=
\left[\begin{array}{cccccc}
\tfrac{3}{2} & 0 & 0 & 0 & 0 & 0 \\
0 & \tfrac{3}{2} & 0 & 0 & 0 & 0 \\
0 & 0 & \tfrac{4}{3} & 0 & 0 & 0 \\
0 & 0 & 0 & \tfrac{4}{3} & 0 & 0 \\
0 & 0 & 0 & 0 & \tfrac{4}{3} & 0 \\
0 & 0 & 0 & 0 & 0 & \tfrac{3}{2}
\end{array}\right]\,.
\end{equation}
Following rotation, 
\begin{equation}\label{eq:tildeB}
\widetilde{B}
:=
\widetilde{A}\,B\widetilde{A}^{\,T}
\,,
\end{equation}
where $\widetilde{A}$ is given in expression~\eqref{eq:At}.
It can be shown by direct calculation that the $(3,3)$ entry of~$\widetilde{B}$ is
\begin{equation}
\label{eq:33}
\widetilde{B}_{33}
=
\frac{4}{3}\left(
1 +
2\left(b^{2} + c^{2}\right)^{2}\left(a^{2} + d^{2}\right)^{2}
\right)
\,.
\end{equation}
Since $C_{1}$ and $C_{2}$ are multiples of the identity,
\begin{equation*}
\widetilde{A}\,C_{1}\widetilde{A}^{\,T} = C_{1}
\quad{\rm and}\quad
\widetilde{A}\,C_{2}\widetilde{A}^{\,T} = C_{2}
\,,
\end{equation*}
and the Backus average of~$\widetilde{A}\,C_{1}\widetilde{A}^{\,T}$ and~$\widetilde{A}\,C_{2}\widetilde{A}^{\,T}$ equals the Backus average of~$C_{1}$ and~$C_{2}$\,, which is matrix~\eqref{eq:BAofC1andC2}\,.
Hence,
\begin{equation*}
\widetilde{A}\,\,\overline{\overline{\left(\,C_{1}\,,\,C_{2}\,\right)}}\,\widetilde{A}^{\,T}
=
\overline{\overline{\left(\,\widetilde{A}\,C_{1}\widetilde{A}^{\,T}\,,\,\widetilde{A}\,C_{2}\widetilde{A}^{\,T}\,\right)}}
\end{equation*}
implies that, for expression~\eqref{eq:33},
\begin{equation*}
\frac{4}{3}\left(
1 +
2\left(b^{2} + c^{2}\right)^{2}\left(a^{2} + d^{2}\right)^{2}
\right)
=
\frac{4}{3}\,,
\end{equation*}
which results in
\begin{equation*}
\frac{2\sqrt{2}}{3}\left(b^{2} + c^{2}\right)\left(a^{2} + d^{2}\right)
=0\,.
\end{equation*}
Thus, either $b=c=0$ or $a=d=0$\,.
This is a necessary condition for symmetry~\eqref{eq:BAsym} to hold, as claimed.
\end{proof}
\begin{remark}\label{rem:DiagMatrices}
Theorem~\ref{thm:SymProp} is formulated for general positive-definite matrices~$C\in\R^{6\times6}$\,, not all of which represent elasticity tensors.
However, expression~\eqref{eq:BAsym} is continuous in the~$C_i$ and hence is true in general only if it is also true for~$C_i,$ such as diagonal matrices, which are limits of elasticity tensors.
\end{remark}
\section{The Backus average of randomly oriented tensors}
In this section, we study the Backus average for a random orientations of a given tensor.
As discussed in Section~\ref{sec:BackusGazis}, the arithmetic average of such orientations results in the Gazis et al. average, which is the closest isotropic tensor with respect to the Frobenius norm.
We see that---for the Backus average---the result is, perhaps surprisingly, different.

Given an elasticity tensor, $C\in\R^{6\times6}$\,, let us consider a sequence of its random rotations given by
\begin{equation*}
C_j :=
\widetilde{Q}_j\,C\,\widetilde{Q}_j^{\,T},\quad j=1\,,\,\ldots\,,\,n\,,
\end{equation*}
where~$Q_j\in\R^{3\times3}$ are random matrices sampled from $SO(3)$\,.

The~$C_j$ are samples from some distribution and, hence, almost surely,
\begin{equation*}
\overline{C}:=\lim_{n\to\infty}\frac{1}{n}\sum_{j=1}^{n} C_{j} =\mu(C)\,,
\end{equation*}
the true mean,
\begin{equation*}
\mu(C)=\int\limits_{SO(3)} \widetilde{Q}\,C\,\widetilde{Q}^{\,T}\,{\rm d}\sigma(Q)
\,,
\end{equation*}
where~${\rm d}\sigma(Q)$ is Haar measure on~$SO(3)$\,.
Note that~$\mu(C)$ is just the Gazis et al. average of~$C$\,.

Similarly, for any expression~$X(C)$ of submatrices of~$C$\,, which appear in the Backus-average formulas,
\begin{equation*}
\overline{X}:=\lim_{n\to\infty}\frac{1}{n}\sum_{j=1}^{n} X_{j} =\mu(X)\,.
\end{equation*}
Hence, almost surely~$\lim_{n\to\infty}\overline{\overline{\left(\,C_{1}\,,\,\ldots\,,\,C_{n}\right)}}$ equals the Backus average formula with each expression~$\overline{X}$ replaced by
\begin{equation*}
\mu(X)=\int\limits_{SO(3)} X\left(\widetilde{Q_j}\,C\,\widetilde{Q_j}{}^{\,T}\right){\rm d}\sigma(Q)\,.
\end{equation*}

\begin{theorem}\label{thm:BTI}
The~$\lim_{n\to\infty}\overline{\overline{\left(\,C_{1}\,,\,\ldots\,,\,C_{n}\right)}}$ exists almost surely, in which case it is transversely isotropic.
It is not, in general, isotropic.
\end{theorem}
\begin{proof}
Let~$A\in\R^{3\times3}$ be an orthogonal matrix of type \eqref{eq:A1} or \eqref{eq:A2}.
Then
\begin{align*}
\widetilde{C}_j
&:=\widetilde{A}\,C_j\widetilde{A}^{\,T},\quad j=1\,,\,\ldots\,,\,n\,,\\
&=\widetilde{A}\left(\widetilde{Q_j}\,C_{j}\,\widetilde{Q_j}{}^{\,T}\right)\widetilde{A}^{\,T}\\
&=\left(\widetilde{A}\,\widetilde{Q_j}\right)C_{j}\left(\widetilde{A}\,\widetilde{Q_j}\right)^{\,T}\\
&=\widetilde{\left(AQ_j\right)}\,C_{j}\,\widetilde{\left(AQ_j\right)}^{\,T}
\end{align*}
by the properties of the tilde operation, are also random samples from the same distribution.
Hence, almost surely,
\begin{equation*}
\lim_{n\to\infty}\overline{\overline{
\left(\,\widetilde{C}_1\,,\,\ldots\,,\,\widetilde{C}_n\,\right)
}}
=
\lim_{n\to\infty}\overline{\overline{
\left(\,C_1\,,\,\ldots\,,\,C_n\,\right)
}}
=
B\,,
\quad\hbox{say}\,.
\end{equation*}
But by the symmetry property of the Backus average, Theorem~\ref{thm:SymProp},
\begin{equation*}
\overline{\overline{
\left(\,\widetilde{C}_1\,,\,\ldots\,,\,\widetilde{C}_n\,\right)
}}
=
\widetilde{A}\,\,\overline{\overline{
\left(\,C_1\,,\,\ldots\,,\,C_n\,\right)
}}\,\widetilde{A}^{\,T}.
\end{equation*}
Thus
\begin{equation*} 
B=
\widetilde{A}\,B\,\widetilde{A}^{\,T}\,,
\end{equation*}
which means that $B$ is invariant under a rotation of space by~$A$\,.
Consequently, $B$ is a transversely isotropic tensor.

In general, the limit tensor is not isotropic, as illustrated by the following example.
Let
\begin{equation*}
C={\rm diag}\left[\,1\,,\,1\,,\,1\,,\,1\,,\,0\,,\,0\,\right]\,,
\end{equation*}
which, as stated in Remark~\ref{rem:DiagMatrices}, represents a limiting case of an elasticity tensor.
Numerical evidence strongly suggests  that 
\begin{equation*}
B
=
\left[\def\arraystretch{1.25}\begin{array}{cccccc}
\tfrac{1}{2} & \tfrac{1}{4} & \tfrac{1}{4} & 0 & 0 & 0 \\
\tfrac{1}{4} & \tfrac{1}{2} & \tfrac{1}{4} & 0 & 0 & 0 \\
\tfrac{1}{4} & \tfrac{1}{4} & \tfrac{1}{2} & 0 & 0 & 0 \\
0 & 0 & 0 & 0 & 0 & 0 \\
0 & 0 & 0 & 0 & 0 & 0 \\
0 & 0 & 0 & 0 & 0 & \tfrac{1}{4}
\end{array}\right]\,,
\end{equation*}
which is {\it not} isotropic.

Although this is rather an artificial example, it could---with some computational effort---be ``promoted'' to a legal proof.
The conclusion is readily confirmed by the numerical examples presented in Section~\ref{sec:NumericalExamples}.
\end{proof}
In fact, it is easy to identify the limiting matrix~$B$\,; it is just the Backus average expression~\eqref{eq:BAforms}, with an expression~$\overline{X(C)}$ replaced by the true mean
\begin{equation}\label{eq:SectionB}
\mu(X(C))=\int\limits_{SO(3)} X\left(\widetilde{Q}\,C\,\widetilde{Q}^{\,T}\right){\rm d}\sigma(Q)\,.
\end{equation}
This limiting transversely isotropic tensor is of natural interest in its own right.
It plays the role of the Gazis et al. average in the context of the Backus average, and is the subject of a forthcoming work.
\section{Numerical example}
\label{sec:NumericalExamples}
Let us consider the elasticity tensor obtained by~\citet{DewanganGrechka2003}; its components are estimated from seismic measurements in New Mexico,
\begin{equation}\label{eq:GrechkaTensor}
C=
\left[\begin{array}{c*{5}{c}}
7.8195 & 3.4495 & 2.5667 & \sqrt{2}\,(0.1374) & \sqrt{2}\,(0.0558) & \sqrt{2}\,(0.1239)\\
3.4495 & 8.1284 & 2.3589 & \sqrt{2}\,(0.0812) & \sqrt{2}\,(0.0735) & \sqrt{2}\,(0.1692)\\
2.5667 & 2.3589 & 7.0908 & \sqrt{2}\,(-0.0092) & \sqrt{2}\,(0.0286) & \sqrt{2}\,(0.1655)\\
\sqrt{2}\,(0.1374) & \sqrt{2}\,(0.0812) & \sqrt{2}\,(-0.0092) & 2\,(1.6636) & 2\,(-0.0787) & 2\,(0.1053)\\
\sqrt{2}\,(0.0558) & \sqrt{2}\,(0.0735) & \sqrt{2}\,(0.0286) & 2\,(-0.0787) & 2\,(2.0660) & 2\,(-0.1517)\\
\sqrt{2}\,(0.1239) & \sqrt{2}\,(0.1692) & \sqrt{2}\,(0.1655) & 2\,(0.1053) & 2\,(-0.1517) & 2\,(2.4270)
\end{array}\right]
\,.
\end{equation}
Using tensor~\eqref{eq:GrechkaTensor}, let us demonstrate two methods to obtain~$B$ and their mutual convergence in the limit.

The first method to obtain~$B$ requires a stack of layers, whose elasticity tensors are $C$\,.
We rotate each~$C$\,, using a random unit quaternion, and perform the Backus average of the resulting stack of layers.
Using~$10^{7}$ layers, the Backus average is
\begin{equation}\label{eq:Brot}
B_{\equiv}
=
\left[\begin{array}{c*{5}{c}}
7.3008 & 
2.9373 & 
2.9379 & 
\sqrt{2}\left(0.0000\right) & 
\sqrt{2}\left(0.0000\right) & 
\sqrt{2}\left(0.0000\right) 
\\
2.9373 & 
7.3010 & 
2.9381 & 
\sqrt{2}\left(0.0000\right) & 
\sqrt{2}\left(0.0000\right) & 
\sqrt{2}\left(0.0000\right) 
\\
2.9379 & 
2.9381 & 
7.2689 & 
\sqrt{2}\left(0.0000\right) & 
\sqrt{2}\left(-0.0001\right) & 
\sqrt{2}\left(0.0000\right) 
\\
\sqrt{2}\left(0.0000\right) & 
\sqrt{2}\left(0.0000\right) & 
\sqrt{2}\left(0.0000\right) & 
2\left(2.1710\right) & 
2\left(0.0000\right) & 
2\left(0.0000\right) 
\\
\sqrt{2}\left(0.0000\right) & 
\sqrt{2}\left(0.0000\right) & 
\sqrt{2}\left(-0.0001\right) & 
2\left(0.0000\right) & 
2\left(2.1710\right) & 
2\left(0.0000\right) 
\\
\sqrt{2}\left(0.0000\right) & 
\sqrt{2}\left(0.0000\right) & 
\sqrt{2}\left(0.0000\right) & 
2\left(0.0000\right) & 
2\left(0.0000\right) & 
2\left(2.1819\right)
\end{array}\right]
\,.
\end{equation}
For an explicit formulation of the Backus average of generally anisotropic media, see~\citet[expressions~(5)--(9)]{BDaltonSS2017}.

The second method requires integrals in place of arithmetic averages.
Similarly to the first method, we use a random unit quaternion, which is tantamount to a point on a 3-sphere.
We approximate the triple integral using Simpson's and trapezoidal rules.
Effectively, the triple integral is replaced by a weighted sum of the integrand evaluated at discrete points. 
The sums that approximate the integrals are accumulated and are used in expressions~\eqref{eq:BAforms}.

Using the Simpson's and trapezoidal rules, with a sufficient number of subintervals, the Backus average is
\begin{equation}\label{eq:Bint}
B_{\!\int\!\!\!\int\!\!\!\int}
=
\left[\begin{array}{c*{5}{c}}
7.3010 & 
2.9373 & 
2.9380 & 
\sqrt{2}\left(0.0000\right) & 
\sqrt{2}\left(0.0000\right) & 
\sqrt{2}\left(0.0000\right) 
\\
2.9373 & 
7.3010 & 
2.9380 & 
\sqrt{2}\left(0.0000\right) & 
\sqrt{2}\left(0.0000\right) & 
\sqrt{2}\left(0.0000\right) 
\\
2.9380 & 
2.9380 & 
7.2687 & 
\sqrt{2}\left(0.0000\right) & 
\sqrt{2}\left(0.0000\right) & 
\sqrt{2}\left(0.0000\right) 
\\
\sqrt{2}\left(0.0000\right) & 
\sqrt{2}\left(0.0000\right) & 
\sqrt{2}\left(0.0000\right) & 
2\left(2.1711\right) & 
2\left(0.0000\right) & 
2\left(0.0000\right) 
\\
\sqrt{2}\left(0.0000\right) & 
\sqrt{2}\left(0.0000\right) & 
\sqrt{2}\left(-0.0001\right) & 
2\left(0.0000\right) & 
2\left(2.1711\right) & 
2\left(0.0000\right) 
\\
\sqrt{2}\left(0.0000\right) & 
\sqrt{2}\left(0.0000\right) & 
\sqrt{2}\left(0.0000\right) & 
2\left(0.0000\right) & 
2\left(0.0000\right) & 
2\left(2.1818\right)
\end{array}\right]
\,.
\end{equation}

In the limit, the components of expressions~\eqref{eq:Brot} and~\eqref{eq:Bint} are the same; their similarity is illustrated in Figure~\ref{fig:Dist}, where the horizontal axis is the number of layers and the vertical axis is the maximum componentwise difference between the two tensors.

\begin{figure}
\centering
\includegraphics[width=0.75\textwidth]{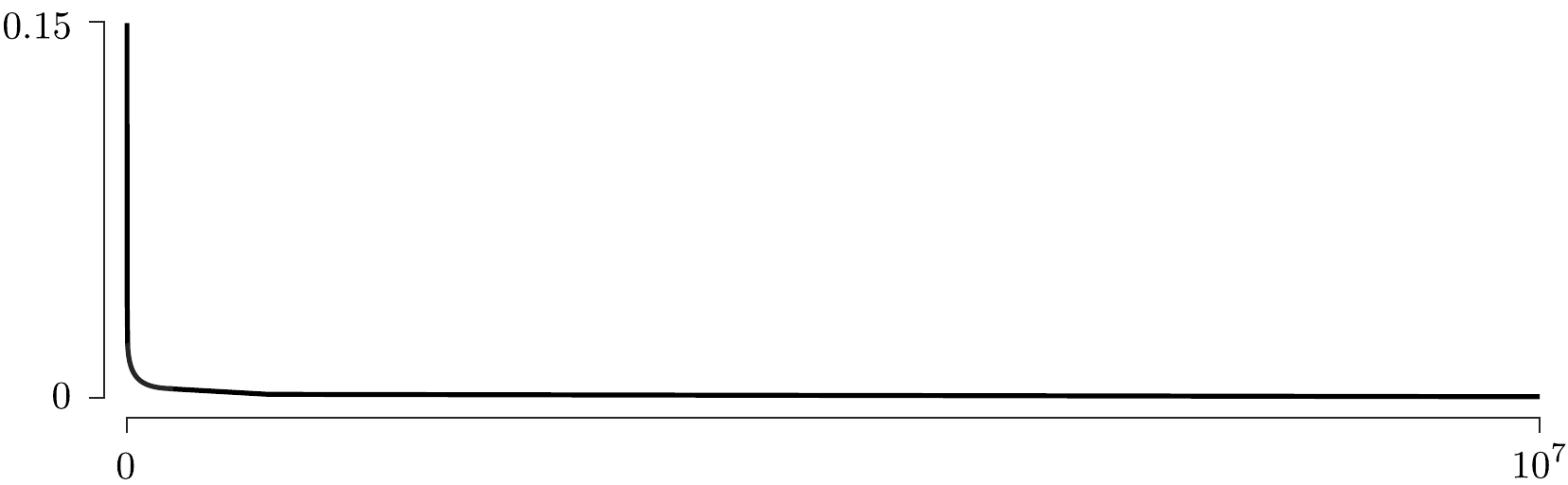}
\caption{
Difference between tensors~\eqref{eq:Brot} and~\eqref{eq:Bint}
}
\label{fig:Dist}
\end{figure}

Expression~\eqref{eq:Bint} is transversely isotropic, as expected from Theorem~\ref{thm:BTI}, and in accordance with \citet[Section~4.3]{BonaEtAl2007}, since its four distinct eigenvalues are
\begin{equation}\label{eq:EigValsBint}
\lambda_{1}=13.1658\,,\,
\lambda_{2}=4.3412\,,\,
\lambda_{3}=4.3636\,,\,
\lambda_{4}=4.3421\,,\,
\end{equation}
with multiplicities of $m_{1}=m_{2}=1$ and $m_{3}=m_{4}=2$\,.
The eigenvalues of expression~\eqref{eq:Brot} are in agreement---up to $10^{-3}$---with eigenvalues~\eqref{eq:EigValsBint} and their multiplicities.
Furthermore, in accordance with Theorem~\ref{thm:BTI}, in the limit, the distance to the closest isotropic tensor for expression~\eqref{eq:Bint} is $0.0326\neq 0$\,; thus the distance does not reduce to zero.

Expressions~\eqref{eq:Brot} and~\eqref{eq:Bint} are transversely isotropic, which is the main conclusion of this work, even though, for numerical modelling, one might view them as isotropic.
This is indicated by~\citet{Thomsen1986} parameters, which for tensor~\eqref{eq:Bint} are
\begin{equation*}
\gamma=2.4768\times10^{-3}\,,\,
\delta=1.5816\times10^{-3}\,,\,
\epsilon=2.2219\times10^{-3}\,; 
\end{equation*}
values much less than unity indicate very weak anisotropy.
\section{Conclusions and future work}
\label{sec:FutureWork}
Examining the Backus average of a stack of layers consisting of randomly oriented anisotropic elasticity tensors, we show that---in the limit---this average results in a homogeneous transversely isotropic medium, as stated by Theorems~\ref{thm:SymProp} and~\ref{thm:BTI}.
In other words, the randomness within layers does not result in a medium lacking a directional pattern.
Both the isotropic layers, as shown by~\citet{Backus1962}, and randomly oriented anisotropic layers, as shown herein, result in the average that is transversely isotropic, as a consequence of inhomogeneity among parallel layers.
This property is discussed by~\citet{AdamusEtAl2018}, and herein it is illustrated in Appendix~\ref{app:Alternating}.

In the limit, the transversely isotropic tensor is the Backus counterpart of the Gazis et al. average.
Indeed, the arithmetic average of randomized layers of an elasticity tensor produces the Gazis et al. average and is its closest isotropic tensor, according to the Frobenius norm.
On the other hand, the Backus average of the layers resulting from a randomization of the same tensor produces the transversely isotropic tensor given in expression~\eqref{eq:SectionB}.
This tensor and its properties are the subject of a forthcoming paper.
\section*{Acknowledgments}
We wish to acknowledge discussions with Michael G. Rochester, proofreading of David R. Dalton, as well as the graphic support of Elena Patarini.
This research was performed in the context of The Geomechanics Project supported by Husky Energy. 
Also, this research was partially supported by the Natural Sciences and Engineering Research Council of Canada, grant 238416-2013.
\bibliographystyle{apa}
\bibliography{BSS_arXiv.bib}
\begin{appendix}
\setcounter{equation}{0}
\renewcommand{\theequation}{\thesection.\arabic{equation}}
\section{Rotations by unit quaternions}
\label{app:Atilde}
The $\mathbb{R}^{6}$ equivalent for $A\in SO(3)$ of $c_{ijk\ell}$\,, which is the rotation of tensor~\eqref{eq:Chapman}, is~\citep[e.g.,][equation~(3.42)]{Slawinski2018}
\begin{align*}
&\nonumber \tilde{A}=\\
&{{\left[\!\begin{array}{c*{5}{c}}
A_{11}^{2} & A_{12}^{2} & A_{13}^{2} & \sqrt{2}\,A_{12}A_{13} & \sqrt{2}\,A_{11}A_{13} & \sqrt{2}\,A_{11}A_{12}\\
A_{21}^{2} & A_{22}^{2} & A_{23}^{2} & \sqrt{2}\,A_{22}A_{23} & \sqrt{2}\,A_{21}A_{23} & \sqrt{2}\,A_{21}A_{22}\\
A_{31}^{2} & A_{32}^{2} & A_{33}^{2} & \sqrt{2}\,A_{32}A_{33} & \sqrt{2}\,A_{31}A_{33} & \sqrt{2}\,A_{31}A_{32}\\
\sqrt{2}\,A_{21}A_{31} & \sqrt{2}\,A_{22}A_{32} & \sqrt{2}\,A_{23}A_{33} & A_{23}A_{32}+A_{22}A_{33} & A_{23}A_{31}+A_{21}A_{33} & A_{22}A_{31}+A_{21}A_{32}\\
\sqrt{2}\,A_{11}A_{31} & \sqrt{2}\,A_{12}A_{32} & \sqrt{2}\,A_{13}A_{33} & A_{13}A_{32}+A_{12}A_{33} & A_{13}A_{31}+A_{11}A_{33} & A_{12}A_{31}+A_{11}A_{32}\\
\sqrt{2}\,A_{11}A_{21} & \sqrt{2}\,A_{12}A_{22} & \sqrt{2}\,A_{13}A_{23} & A_{13}A_{22}+A_{12}A_{23} & A_{13}A_{21}+A_{11}A_{23} & A_{12}A_{21}+A_{11}A_{22}
\end{array}\right]}}
\,.
\end{align*}
In quaternions, this expression is
\begin{align}\label{eq:At}
\nonumber\widetilde{A} &= \\
\nonumber&\left[\small\begin{array}{cc}
\left(a^{2}+b^{2}-c^{2}-d^{2}\right)^{2} & \left(2\,b\,c - 2\,a\,d\right)^{2} \\
\left(2\,b\,c + 2\,a\,d\right)^{2} & \left(a^{2}-b^{2}+c^{2}-d^{2}\right)^{2} \\
\left(2\,b\,d - 2\,a\,c\right)^{2} & \left(2\,a\,b + 2\,c\,d\right)^{2} \\
\sqrt{2}\left(2\,b\,c + 2\,a\,d\right)\left(2\,b\,d - 2\,a\,c\right) & \sqrt{2}\left(2\,a\,b + 2\,c\,d\right)\left(a^{2}-b^{2}+c^{2}-d^{2}\right) \\
\sqrt{2}\left(2\,b\,d - 2\,a\,c\right)\left(a^{2}+b^{2}-c^{2}-d^{2}\right) & \sqrt{2}\left(2\,b\,c - 2\,a\,d\right)\left(2\,a\,b + 2\,c\,d\right) \\
\sqrt{2}\left(2\,b\,c + 2\,a\,d\right)\left(a^{2}+b^{2}-c^{2}-d^{2}\right) & \sqrt{2}\left(2\,b\,c - 2\,a\,d\right)\left(a^{2}-b^{2}+c^{2}-d^{2}\right)
\end{array}\right.\\
\nonumber&\left.\small\hspace*{-0.25cm}\begin{array}{cc}
\left(2\,a\,c + 2\,b\,d\right)^{2} & \sqrt{2}\left(2\,b\,c - 2\,a\,d\right)\left(2\,a\,c + 2\,b\,d\right) \\
\left(2\,c\,d - 2\,a\,b\right)^{2} & \sqrt{2}\left(2\,c\,d - 2\,a\,b\right)\left(a^{2}-b^{2}+c^{2}-d^{2}\right) \\
\left(a^{2}-b^{2}-c^{2}+d^{2}\right)^{2} & \sqrt{2}\left(2\,a\,b + 2\,c\,d\right)\left(a^{2}-b^{2}-c^{2}+d^{2}\right) \\
\sqrt{2}\left(2\,c\,d - 2\,a\,b\right)\left(a^{2}-b^{2}-c^{2}+d^{2}\right) & \left(2\,c\,d - 2\,a\,b\right)\left(2\,a\,b + 2\,c\,d\right) + \left(a^{2}-b^{2}+c^{2}-d^{2}\right)\left(a^{2}-b^{2}-c^{2}+d^{2}\right) \\
\sqrt{2}\left(2\,a\,c + 2\,b\,d\right)\left(a^{2}-b^{2}-c^{2}+d^{2}\right) & \left(2\,a\,c + 2\,b\,d\right)\left(2\,a\,b + 2\,c\,d\right) + \left(2\,b\,c - 2\,a\,d\right)\left(a^{2}-b^{2}-c^{2}+d^{2}\right) \\
\sqrt{2}\left(2\,a\,c + 2\,b\,d\right)\left(2\,c\,d - 2\,a\,b\right) & \left(2\,c\,d - 2\,a\,b\right)\left(2\,b\,c - 2\,a\,d\right) + \left(2\,a\,c + 2\,b\,d\right)\left(a^{2}-b^{2}+c^{2}-d^{2}\right) \\
\end{array}\right.\\
\nonumber&\left.\small\hspace{2.75cm}\begin{array}{c}
\sqrt{2}\left(2\,a\,c + 2\,b\,d\right)\left(a^{2}+b^{2}-c^{2}-d^{2}\right) \\
\sqrt{2}\left(2\,b\,c + 2\,a\,d\right)\left(2\,c\,d - 2\,a\,b\right) \\
\sqrt{2}\left(2\,b\,d - 2\,a\,c\right)\left(a^{2}-b^{2}-c^{2}+d^{2}\right) \\
\left(2\,b\,d - 2\,a\,c\right)\left(2\,c\,d - 2\,a\,b\right) + \left(2\,b\,c + 2\,a\,d\right)\left(a^{2}-b^{2}-c^{2}+d^{2}\right) \\
\left(2\,b\,d - 2\,a\,c\right)\left(2\,a\,c + 2\,b\,d\right) + \left(a^{2}+b^{2}-c^{2}-d^{2}\right)\left(a^{2}-b^{2}-c^{2}+d^{2}\right) \\
\left(2\,b\,c + 2\,a\,d\right)\left(2\,a\,c + 2\,b\,d\right) + \left(2\,c\,d-2\,a\,b\right)\left(a^{2}+b^{2}-c^{2}-d^{2}\right)
\end{array}\right.\\
&\left.\small\hspace{4.75cm}\begin{array}{c}
\sqrt{2}\left(2\,b\,c - 2\,a\,d\right)\left(a^{2}+b^{2}-c^{2}-d^{2}\right) \\
\sqrt{2}\left(2\,b\,c + 2\,a\,d\right)\left(a^{2}-b^{2}+c^{2}-d^{2}\right) \\
\sqrt{2}\left(2\,b\,d - 2\,a\,c\right)\left(2\,a\,b + 2\,c\,d\right) \\
\left(2\,b\,c + 2\,a\,d\right)\left(2\,a\,b + 2\,c\,d\right) + \left(2\,b\,d - 2\,a\,c\right)\left(a^{2}-b^{2}+c^{2}-d^{2}\right) \\
\left(2\,b\,d - 2\,a\,c\right)\left(2\,b\,c - 2\,a\,d\right) + \left(2\,a\,b + 2\,c\,d\right)\left(a^{2}+b^{2}-c^{2}-d^{2}\right) \\
\left(2\,b\,c + 2\,a\,d\right)\left(2\,b\,c - 2\,a\,d\right) + \left(a^{2}+b^{2}-c^{2}-d^{2}\right)\left(a^{2}-b^{2}+c^{2}-d^{2}\right)
\end{array}\right]\,.
\end{align}
\section{Alternating layers}
\label{app:Alternating}
Consider a randomly-generated elasticity tensor,
\begin{equation}\label{eq:Alternating}
C=
\left[\begin{array}{c*{5}{c}}
14.5739 & 6.3696 & 2.9020 & \sqrt{2}\,(9.4209) & \sqrt{2}\,(3.8313) & \sqrt{2}\,(3.5851)\\
6.3696 & 10.7276 & 6.2052 & \sqrt{2}\,(4.0375) & \sqrt{2}\,(5.1333) & \sqrt{2}\,(6.0745)\\
2.9020 & 6.2052 & 11.4284 & \sqrt{2}\,(1.9261) & \sqrt{2}\,(9.8216) & \sqrt{2}\,(1.3827)\\
\sqrt{2}\,(9.4209) & \sqrt{2}\,(4.0375) & \sqrt{2}\,(1.9261) & 2\,(13.9034) & 2\,(0.2395) & 2\,(2.0118)\\
\sqrt{2}\,(3.8313) & \sqrt{2}\,(5.1333) & \sqrt{2}\,(9.8216) & 2\,(0.2395) & 2\,(10.7353) & 2\,(0.0414)\\
\sqrt{2}\,(3.5851) & \sqrt{2}\,(6.0745) & \sqrt{2}\,(1.3827) & 2\,(2.0118) & 2\,(0.0414) & 2\,(9.0713)
\end{array}\right]
\,,
\end{equation}
whose eigenvalues are
\begin{equation*}
\lambda_{1}=34.0318\,,\,
\lambda_{2}=18.1961\,,\,
\lambda_{3}=10.4521\,,\,
\lambda_{4}=4.8941\,,\,
\lambda_{5}=2.2737\,,\,
\lambda_{6}=0.5921\,.
\end{equation*}
The Backus average of $10^7$ alternating layers composed of randomly oriented tensors~\eqref{eq:GrechkaTensor} and~\eqref{eq:Alternating} is
\begin{equation*}
B_{\!\int\!\!\!\int\!\!\!\int}=
\left[\begin{array}{c*{5}{c}}
8.4711 & 1.1917 & 1.2572 & \sqrt{2}\,(0.0000) & \sqrt{2}\,(0.0000) & \sqrt{2}\,(0.0000) \\
1.1917 & 8.4710 & 1.2570 & \sqrt{2}\,(0.0000) & \sqrt{2}\,(0.0000) & \sqrt{2}\,(0.0000) \\
1.2572 & 1.2570 & 6.6648 & \sqrt{2}\,(-0.0001) & \sqrt{2}\,(0.0000) & \sqrt{2}\,(0.0000) \\
\sqrt{2}\,(0.0000) & \sqrt{2}\,(0.0000) & \sqrt{2}\,(-0.0001) & 2\,(2.8440) & 2\,(0.0000) & 2\,(0.0000) \\
\sqrt{2}\,(0.0000) & \sqrt{2}\,(0.0000) & \sqrt{2}\,(0.0000) & 2\,(0.0000) & 2\,(2.8440) & 2\,(0.0000) \\
\sqrt{2}\,(0.0000) & \sqrt{2}\,(0.0000) & \sqrt{2}\,(0.0000) & 2\,(0.0000) & 2\,(0.0000) & 2\,(3.6340)
\end{array}\right]
\,.
\end{equation*}
Its eigenvalues show that this is a transversely isotropic tensor,
\begin{equation*}
\lambda_{1}=10.4892\,,\,
\lambda_{2}=5.8384\,,\,
\lambda_{3}=7.2794\,,\,
\lambda_{4}=7.2793\,,\,
\lambda_{5}=5.6880\,,\,
\lambda_{6}=5.6878\,.
\end{equation*}
Its Thomsen parameters,
\begin{equation*}
\gamma=0.1400\,,\,
\delta=0.0433\,,\,
\epsilon=0.1353\,,
\end{equation*}
indicate greater anisotropy than for tensor~\eqref{eq:Bint}, as expected.
In other words, an emphasis of a pattern of inhomogeneity results in an increase of anisotropy.
\end{appendix}
\end{document}